%% file: main.tex
\documentclass[11pt]{article}
	
	\newcommand{\blind}{0}
	
    \makeatletter
    \renewcommand\section{\@startsection {section}{1}{\z@}%
                                       {-3.5ex \@plus -1ex \@minus -.2ex}%
                                       {2.3ex \@plus.2ex}%
                                       {\normalfont\fontfamily{phv}\fontsize{16}{19}\bfseries}}
    \renewcommand\subsection{\@startsection{subsection}{2}{\z@}%
                                         {-3.25ex\@plus -1ex \@minus -.2ex}%
                                         {1.5ex \@plus .2ex}%
                                         {\normalfont\fontfamily{phv}\fontsize{14}{17}\bfseries}}
    \renewcommand\subsubsection{\@startsection{subsubsection}{3}{\z@}%
                                        {-3.25ex\@plus -1ex \@minus -.2ex}%
                                         {1.5ex \@plus .2ex}%
                                         {\normalfont\normalsize\fontfamily{phv}\fontsize{14}{17}\selectfont}}
    \makeatother
	
	\usepackage{amsmath}
	\usepackage{graphicx}
	\usepackage{enumerate}
	\usepackage{xcolor}
	\usepackage{natbib} 
	\usepackage{url} 
	
	\usepackage[colorlinks, allcolors=blue]{hyperref}
	\usepackage{amsfonts, amsthm, latexsym, amssymb}
	\usepackage{booktabs}
	\usepackage{siunitx}
	\newtheorem{prop}{Proposition}
	\newtheorem{remark}[prop]{Remark}
	\newtheorem{lemma}[prop]{Lemma}
	
\begin{document}
		
		\def\spacingset#1{\renewcommand{\baselinestretch}%
			{#1}\small\normalsize} \spacingset{1}
		
		\if0\blind
		{
			\title{Sensitivity Analysis and Optimization of Stochastic Epidemic Models under Parameter Uncertainty}
			\author{Nicholas R. Wu$^a$, Michael C. Fu$^b,^c$ \\
			$^a$ Department of Mathematics, University of Maryland, College Park, MD, USA\\
             	        $^b$ Robert H. Smith School of Business, University of Maryland, College Park, MD, USA\\
		        $^c$ Institute for Systems Research, University of Maryland, College Park, MD, USA
     		}
			\date{}
			\maketitle
		} \fi
		\if1\blind
		{
			\title{Sensitivity Analysis and Optimization of Stochastic Epidemic Models under Parameter Uncertainty}
			\date{}
			\maketitle
			\vspace{-0.5in}
		} \fi


		%

		
	\input{abstract.tex}		

\input{intro.tex}

\input{sir_models.tex}

\input{gradients.tex}

\input{parameter_uncert.tex}

\input{numerical_sec_intro.tex}

\input{sensitivity_results.tex}

\input{opt_results.tex}

\input{conclusion.tex}

\textbf{Data Availability Statement:} The data and code that support the findings of this study are available from the corresponding author upon reasonable request.

\textbf{Use of Generative AI Tools:} ChatGPT-5.3 was used during the revision and writing process to suggest alternate phrasings and improve grammar. Gemini-2.5 was used for formatting assistance, coding assistance, and assistance with the proofs. 

\textbf{Disclosure Statement:} No potential conflict of interest was reported by the author(s).

\bibliographystyle{chicago}
\spacingset{1}
\bibliography{litreview, ref}

%

\end{document}

%% file: abstract.tex
\begin{abstract}
	To address sensitivity analysis and optimization for a discrete-time stochastic epidemic model, we derive unbiased gradient estimators that accommodate uncertainties represented as distributions over the parameters of interest, such as those arising from Bayesian calibration.
	Specifically, we estimate the sensitivity of total infections over a finite time horizon with respect to the proportion immunized ($v$) and the contact rate ($\beta$). Comparing the proposed estimators with deterministic limit approximations based on large populations reveals differences due to the finite population and time horizon. 
The estimators exhibit lower variance than finite-difference estimators for the derivative with respect to $\beta$, but higher variance for the derivative with respect to $v$.
Simulation experiments indicate parameter uncertainty reduces sensitivity to the parameters of interest. In particular, indirect effects of vaccination, such as herd immunity, are less pronounced compared to when parameters are known. For optimization problems balancing intervention and infection costs, incorporating parametric uncertainty leads to more conservative policies. 

\end{abstract}

%% file: intro.tex
\vspace{-6pt}
\section{Introduction} \label{sec:intro}
Mathematical models of infectious diseases are routinely used to inform health policy, with the modeling of uncertainty playing a major role, since disease transmission dynamics are inherently stochastic. Additionally, there may be limited data for estimating parameters, particularly during the early stages of an epidemic. It is therefore important to develop modeling techniques that can cope with these different sources of uncertainty.
\cite{saltelli2020five} advocated for the use of uncertainty quantification and sensitivity analysis techniques to help ensure that important policy decisions are not based upon spurious predictions. \cite{taleb2022single} and \cite{ioannidis2022forecasting} highlighted different limitations of models during the COVID-19 pandemic, but both emphasized the shortcomings of point estimates and advocated for probabilistic approaches to epidemic forecasting.

In particular, an important problem in epidemic modeling is to understand how uncertainty affects policy decisions. 
Different approaches have distinct tradeoffs for addressing this problem, as the assumptions made may or may not be appropriate in a given setting. For example, deterministic Susceptible-Infected-Recovered (SIR) models (\cite{hethcote2000mathematics}) are computationally tractable and can readily incorporate parameter uncertainty through standard uncertainty quantification methods, but they neglect stochasticity in transmission dynamics and may be less appropriate for small populations and short time periods. 
Agent-based models can capture individual-level stochastic transmission, and have recently been combined with differentiable programming techniques for end-to-end calibration, sensitivity analysis, and optimization (\cite{chopra2022differentiable}), but these models often require high-resolution data for model calibration, which may not be available (\cite{iranzo2021epidemiological}). Stochastic SIR models (\cite{allen2008introduction}) offer an alternative approach, capturing randomness in transmission dynamics while operating at the same resolution as deterministic SIR models. Despite this, there has been comparatively little work on sensitivity analysis and optimization under uncertainty for such models (\cite{swallow2022challenges}). 

Motivated by these considerations, we study stochastic SIR models for analyzing uncertainty in epidemic dynamics. We develop gradient-based methods for sensitivity analysis and optimization that account for uncertainty in key parameters. Namely, we focus on deriving gradient estimators for the total infections over a finite time horizon with respect to the proportion immunized $(v)$ and the effective contact rate $(\beta)$, where $\beta$ accounts for both the rate of contact between individuals and the probability that such a contact results in a new infection. To highlight potential implications for health policy, we illustrate how uncertainty impacts the effectiveness of interventions (specifically, vaccinations and contact reductions) through simulation experiments.

In sum, this work demonstrates the utility of gradient estimation for analyzing stochastic epidemic models via the following research contributions:

\begin{enumerate}
    \item We develop unbiased stochastic gradient estimators for a discrete-time stochastic SIR model. We extend these estimators to incorporate parameter uncertainty (such as distributions arising from Bayesian calibration) within a simulation-based framework.

    \item After numerically validating the proposed estimators against finite-difference estimators and deterministic formulas derived from the large-population limit of epidemic models, we show that the proposed estimators admit intuitive interpretations and provide insight into epidemic dynamics, including indirect effects such as herd immunity. 
	
    \item We investigate how parameter uncertainty affects the marginal impact of interventions by applying the gradient estimators to sensitivity analysis. Simulation experiments indicate that when averaging across plausible parameter values, the expected impact of an intervention is reduced.

    \item  
	We investigate how parameter uncertainty influences the optimal amount of intervention by solving a stylized optimization problem that models tradeoffs between intervention costs and infection costs. The results indicate that accounting for parameter uncertainty leads to more vaccination and more contact reduction, compared to no parameter uncertainty.
\end{enumerate}

The rest of the paper is organized as follows. In the remaining portions of Section \ref{sec:intro}, we discuss the related literature and describe the problem setup. In Section \ref{sec:sir_models}, we review classical results for the continuous-time stochastic SIR model, including deterministic limit approximations for large populations, and introduce the discrete-time approximate stochastic SIR model used for simulation. In Section \ref{sec:grad_est}, we derive the stochastic gradient estimators and discuss variance reduction techniques. In Section \ref{sec:parameter_uncertainty}, we define sensitivities under parameter uncertainty, and show how the gradient estimators can be used to estimate such sensitivities. In Section \ref{sec:numerical}, we present simulation experiments illustrating the proposed methods, including the numerical validation, sensitivity analysis, and optimization problems.

\vspace{-6pt}
\subsection{
\emph{Related literature}} \label{s:lit_review}

\vspace{-6pt}
\subsubsection{Sensitivity analysis of epidemic models}

The goal of sensitivity analysis is to help the analyst understand which inputs of the system contribute most to variations in the outcome of interest. Parameters that are less influential may be fixed or removed, whereas more effort should be spent on estimating parameters that are highly influential. Sensitivity analysis approaches are divided into local and global methods (\cite{saltelli2008global}). Local methods characterize how perturbations around a single point in the parameter space affect the output of interest. An example of a local sensitivity measure is the partial derivative. In contrast, global methods aim to describe the influence of parameters across the entire parameter space. 
In this work, we focus on local sensitivity analysis via gradient estimation.

Computing the sensitivity measure may require knowledge of the structure of the model, or may treat the model as a black-box. Deterministic SIR models are amenable to local gradient-based approaches through direct derivation or adjoint methods (\cite{berhe2019parameter}, \cite{mester2022differential}), as well as black-box global methods (\cite{wu2013sensitivity}). Agent-based models typically are analyzed using black-box approaches, either via direct Monte-Carlo methods or metamodeling methods (\cite{nsoesie2012sensitivity}, \cite{borgonovo2022sensitivity}). \cite{kouye2022sensitivity} developed global sensitivity analysis methods for SIR-type stochastic epidemic models by analyzing different constructions of the SIR model. \cite{lu2023global} applied black-box approaches to study a stochastic epidemic model. However, there is little work on local sensitivity methods for stochastic models.

\vspace{-6pt}
\subsubsection{Optimization of epidemic models}
Optimization routines are used with epidemic models for parameter estimation and prescriptive analyses. 
Prior work has investigated gradient-based methods for the calibration of stochastic compartmental epidemic models. \cite{alaeddini2017application} used a black-box gradient estimation algorithm to fit an epidemic model to data. \cite{rupp2024differentiated} developed a numerical linear-algebraic approach to calculate derivatives of the state distribution's likelihood with respect to the parameters of a stochastic SIR model. 

Optimization routines are also used extensively for prescriptive analyses utilizing epidemic models. Often, the goal is to minimize costs (e.g., healthcare system burdens or infections) subject to resource constraints. Some model-agnostic simulation optimization methodologies used include metaheuristics (\cite{gillis2021simulation}), the Nelder-Mead simplex method (\cite{nsoesie2013simulation}) and ranking-and-selection procedures (\cite{paleshi2017simulation}). Another approach is to integrate the epidemic model with a mathematical programming model (\cite{yin2022risk}). 
When using this approach, the epidemic model usually must be simplified so that the resulting optimization model is tractable. Alternatively, a sampling-based approach can be used, and a finite set of scenarios can be generated from the simulation model.
Prior work has also used gradient-based optimization for deterministic models (\cite{navascues2021disease}). Relevant to the current work is a growing interest in applying automatic differentiation tools to agent-based models, which automatically compute gradient estimates and yield end-to-end differentiability for both calibration and optimization (\cite{chopra2022differentiable}). This work complements existing approaches based on agent-based models by focusing instead on stochastic SIR models, which may be more appropriate when high-resolution data are unavailable.

\vspace{-6pt}
\subsubsection{Gradient estimation}

The goal of gradient estimation is to develop Monte-Carlo methods to estimate the gradients of simulation outputs with respect to system parameters. 
Broadly speaking, gradient estimation methods may be categorized as indirect or direct (\cite{fu2014stochastic}). An example of an indirect method is a finite-difference estimator, which approximates derivatives by perturbing the input parameters. These methods treat the simulation model as a black box and do not require modification of the simulation, but typically require multiple simulation runs per replication of the estimator, and introduce bias in the gradient estimate due to the finite perturbation size (\cite{fu2014stochastic}).

On the other hand, direct gradient estimation methods utilize the structure of the underlying system and can require modification of the simulation. Provided that regularity conditions hold, direct methods are unbiased; they also offer computational advantages compared to indirect methods. In this work, we focus on two direct approaches: the weak-derivative method and the likelihood-ratio method. The weak-derivative method decomposes the derivative of the probability distribution into two new distributions, generally requiring two simulations to produce one replicate of the gradient estimator \cite[Chapter 3]{pflug2012optimization}. On the other hand, the likelihood-ratio method reweights simulation outputs by the score function of the underlying distribution and requires only one simulation to produce one replicate of the gradient estimator (\cite{glynn1990likelihood}).

\vspace{-6pt}
\subsection{Problem setting}

Consider a population of $N$ individuals. At time $t=0$, a fixed number $i_0$ are infected, and a fixed proportion $v\in [0,1]$ are immune. Let $\beta$ denote the (effective) contact rate, which represents the average contact rate multiplied by the probability the contact results in an infection. $\beta$ and $v$ are the parameters of interest. For example, changes to $\beta$ and $v$ represent policy decisions corresponding to distancing/masking orders and immunization, respectively. Let the random variable $Z(v, \beta ,T)$ denote the total number of individuals that become infected over a time horizon $T$. The quantity of interest is the expected total number of infections, which we write as $J(v,\beta,T) = \mathbb{E}[Z(v,\beta,T)]$. The goal of gradient estimation is to estimate the partial derivatives $J_v$ and $J_\beta$ via a Monte-Carlo procedure. Uncertainty quantification is addressed by assuming a probability distribution over the parameters $v$ and $\beta$.

%% file: sir_models.tex
\vspace{-6pt}
\section{
\emph{Stochastic SIR models}} \label{sec:sir_models}
In Section \ref{sec:sir_ctmc}, we describe the continuous-time Markov chain (CTMC) SIR model and discuss fundamental results such as the threshold limit theorem, which gives conditions for when large outbreaks are possible. Then, in Section \ref{sec:final_size_equations}, we explain how approximate formulas for $J_v$ and $J_\beta$ can be derived based on deterministic large-population limits of the CTMC SIR model. 
For the finite-population case, stochastic gradient estimators for $J_v$ and $J_\beta$ are derived in Section \ref{sec:grad_est}, using a discrete-time approximation to the CTMC SIR model based on the TSIR (Time-Series SIR) model introduced by \cite{grenfell2002dynamics}, and described in Section \ref{sec:discrete_time_approx}.

\vspace{-6pt}
\subsection{Continuous-time Markov chain SIR model} \label{sec:sir_ctmc}

Consider the continuous-time Markov chain (CTMC) formulation of the SIR model. Let $S(t), I(t), R(t)$ denote the susceptible, infected and recovered individuals at time $t$. We assume that the population is closed, i.e., for all times $t$, we have $S(t) + I(t) + R(t) = N$. It then suffices to take the state space as $\mathcal{S} = \{(s,i) \in \mathbb{Z}_{\geq 0} \times \mathbb{Z}_{\geq 0}: s+i\leq N\}$, where $\mathbb{Z}_{\geq 0}$ denotes the nonnegative integers. We first consider an initial state where there are $i_0$ initially infected individuals, and all remaining are susceptible, i.e., 
$S(0) = N - i_0$ and $I(0) = i_0$.
In this model, two types of events occur: (1) $S \rightarrow I$, one susceptible individual is infected, or (2) $I \rightarrow R$, one infected individual recovers. The time between events is exponentially distributed. Let $\lambda(s,i)$ denote the (state-dependent) rate of infection events as a function of $(s,i)$, and similarly let $\mu(s,i)$ denote the rate of recovery events. The rates are defined as
$\lambda(s,i) = \beta i \frac{s}{N}$, and
$\mu(s,i) = \gamma i$
where $\beta$ is called the (effective) contact rate, and $\gamma$ is called the recovery rate. 
This system does not have a unique stationary distribution: each state of the form $(s,0)$ is absorbing, because no new infections can occur, and each state of the form $(s,i)$ for $i > 0$ is transient, because all infectious individuals recover. In other words, the disease always goes extinct after infecting a random proportion of the population.

A central quantity for understanding the system is the basic reproductive number ${\mathcal{R}_0 = \beta/\gamma}$, the product of the effective contact rate and the mean infectious duration. $\mathcal{R}_0$ can be used to understand whether an outbreak will occur in a sufficiently large population. In particular, if $\mathcal{R}_0 \leq 1$, an outbreak cannot occur, and if $\mathcal{R}_0 > 1$, an outbreak will occur with some probability. Moreover, if an outbreak does occur, we can use $\mathcal{R}_0$ to calculate the proportion of the population infected. This is summarized by the threshold limit theorem, described in Theorems 3.1 and 4.2 of \cite{Andersson_Britton_2000}. Theorem 3.1 shows that for any fixed time $t_0$, $I(t_0)$ converges almost surely to $\mathcal{I}(t_0)$ as $N\rightarrow \infty$, the number alive at time $t_0$ in a branching process coupled to the SIR model. An outbreak corresponds to the branching process's population growing to infinity, and no outbreak corresponds to extinction of the branching process. The probability of extinction $q$ is computed by solving for the smallest root of the equation $q = \phi(q)$, where $\phi$ is the probability generating function of the $\text{Poisson}(\mathcal{R}_0)$ distribution. Theorem 4.2 concerns the size of an outbreak, and shows that as $N \rightarrow \infty$, the expected proportion infected at extinction converges to the solution $\tau$ of the equation $1 - e^{-\mathcal{R}_0 \tau} = \tau$. 

In this work, the initial number of immune individuals follows a binomial distribution with success probability $v$, and is independent of the infection dynamics. Letting $V$ denote the number of initially immune individuals, 
\begin{equation}
V \sim \text{Bin}(N - i_0, v), \quad
S(0) = N -i_0 - V, \quad
I(0) = i_0.
\label{eqn:vacc_init_cond}
\end{equation}
To account for vaccination in the threshold limit theorem, we define the effective reproductive number as $\mathcal{R}_{\text{eff}} = (1-v) \mathcal{R}_0$. Rearranging the inequality $\mathcal{R}_{\text{eff}} \leq 1$ yields the inequalities
$v \geq 1 - \gamma/\beta$
and
$\beta \leq \gamma/(1-v),$
i.e., given a fixed $v$, the maximum allowable contact rate to prevent an outbreak; and given a fixed $\beta$, the minimum vaccine coverage to prevent an outbreak. We refer to these inequalities as \textit{threshold conditions.} Without loss of generality, we scale the contact rates $\beta$ and $\gamma$ so that $\gamma = 1$. Then, one time unit corresponds to the mean infectious duration $1/\gamma$. In this work, we follow this convention due to the time scaling in the discrete-time approximation (to be discussed in Section \ref{sec:discrete_time_approx}).

\vspace{-6pt}
\subsection{Sensitivity analysis based upon the final size equations} \label{sec:final_size_equations}

Theorem 4.2 of \cite{Andersson_Britton_2000} yields a deterministic limit approximation for the total number of individuals that are infected. This approximation applies when $\mathcal{R}_0 > 1$ and the population is large, corresponding to the regime where $N \rightarrow \infty$ and $i_0$ is held constant. We denote the (expected) total number infected up to the extinction time of the disease by $Z_\infty$. Note that the extinction of the disease occurs at a random time; we will compare this with the gradient estimation approach developed in Section \ref{sec:grad_est}, which concerns a finite population $N$ up to a fixed time $T$.

\begin{remark}[Final size approximation and sensitivities]
Consider a large-population SIR model with $\gamma = 1$ and $\mathcal{R}_{\mathrm{eff}} = \beta(1-v)$. Let $q$ denote the extinction probability and $\tau$ the final proportion infected. Let $Z_\infty$ denote the expected total number infected up to extinction of the disease. Then,
$Z_\infty \approx (1-q)N\tau.$
We can thus write the derivatives of interest as
$$
\frac{dZ_\infty}{dv} = N\left[-\frac{dq}{dv}\,\tau + (1-q)\frac{d\tau}{dv}\right], \quad
\frac{dZ_\infty}{d\beta} = N\left[-\frac{dq}{d\beta}\,\tau + (1-q)\frac{d\tau}{d\beta}\right].
$$
\end{remark}

To explain the expressions above, consider a population where $N$ is large. If $\mathcal{R}_0 >1$, then $1-q$ corresponds to the probability of a large outbreak occurring. Since $\tau$ is the expected proportion infected, $N\tau$ is the expected number of individuals infected. Thus, we expect that $\mathbb{E}[Z_\infty] \approx (1-q)N\tau$. To compute $q$ considering the effects of vaccination, we plug $\mathcal{R}_\text{eff}$ into the probability generating function for $\text{Poisson}(\mathcal{R}_\text{eff})$, yielding the equation $q = e^{\mathcal{R}_{\text{eff}} (q-1)}$.
To compute $\tau$ considering the effects of vaccination, we solve the equation $\tau = (1-v)(1-e^{- \beta \tau})$ (\cite{Miller_2012}). We refer to these equations as the final size equations, and they can be solved by using a standard root-finding method. 
Then, the derivatives of this expression with respect to $v$ and $\beta$ can be calculated via implicit differentiation. This yields the following formulas, given that $v$ and $\tau$ are known:
$$
\frac{d \tau}{dv}
=
\frac{-(1- e^{-\beta \tau})}
{1 - (1-v)\beta e^{-\beta \tau}},
\quad
\frac{d \tau}{d \beta}
=
\frac{\tau(1-v)e^{-\beta \tau}}
{1 - (1-v)\beta e^{-\beta \tau}},
\quad
\frac{dq}{dv}
=
-\beta \frac{dq}{d \mathcal{R}_{\text{eff}}},
$$
$$
\frac{dq}{d\beta}
=
(1-v)\frac{dq}{d \mathcal{R}_{\text{eff}}},
\quad
\frac{d q}{d \mathcal{R}_{\text{eff}}}
=
\frac{-(q-1)\exp\{\mathcal{R}_{\text{eff}} (q-1)\}}
{\mathcal{R}_{\text{eff}}\exp\{\mathcal{R}_{\text{eff}} (q-1)\} -1}.
$$


\vspace{-6pt}
\subsection{Discrete-time approximation} \label{sec:discrete_time_approx}

The discrete-time approximation is based upon the Time-Series SIR (TSIR) model, introduced by \cite{grenfell2002dynamics}. We provide a brief description of its derivation; more details can be found in \cite{lomeli2021statistical} and 
\cite{Wakefield_Dong_Minin_2019}. First, we discretize time into increments $\{\Delta t, 2\Delta t, 3\Delta t, \ldots\}$ such that $\Delta t = 1/\gamma$, the mean infectious duration. However, since we assume $\gamma=1$, we define the index set of the process to be $\{0, 1, 2, \ldots\}$. Accordingly, we assume all individuals in the $I$ compartment transition to the $R$ compartment at the end of one timestep. We denote the discrete-time process with subscripts $S_t, I_t, R_t$. Due to these assumptions, the balance equations
$S_{t+1} = S_t - I_{t+1},\ 
 R_{t+1} = R_t + I_t$
are satisfied.
Second, we assume that the number of susceptible individuals is approximately constant over the next time interval $[t, t + \Delta t]$. Then the number of infected individuals over the next time interval can be approximated by a continuous-time linear pure-birth process. For a small time perturbation $\delta t$,
$$
\mathbb{P}(I(t+\delta t)=j \mid I(t)=i,S(t)=s) = \begin{cases}
    (\beta \frac{s}{N}) i + o(\delta t) &j=i+1, \\
    o(\delta t) & \text{otherwise.}
\end{cases}
$$
It can be shown that $I_{t+1}$, the total number of new infections over $[t,t+\Delta t]$ is negative-binomial distributed (e.g., see \cite{serfozo2009basics}, Exercise 4.9), with mean equal to
\[
E[I_{t+1} \mid I_t =i_t, S_t =s] = i_t \left(\exp \left\{  \frac{\beta s}{N} \Delta t \right\} -1\right),
\]
and shape equal to $i_t$. 
Taking the first-order Taylor expansion, and since $\Delta t=1$,
\[
i_t \left(\exp \left\{  \frac{\beta s}{N} \right\} -1\right) \approx \beta \frac{s i_t}{N} = \lambda(s,i_t).
\]
We right-truncate the distribution to prevent the number of new infections from exceeding the number of susceptible individuals. Altogether, this gives the following stochastic recursion that can be used to simulate the chain:
\begin{equation}
\begin{gathered}
Y_{t+1} \mid S_t, I_t \sim \begin{cases}
\text{NegBin}\left(\frac{\beta S_t I_t}{N}, I_t\right) & I_t > 0, \\
\delta_0 & I_t = 0,
\end{cases} \\
I_{t+1} = \min(S_t, Y_t), \quad 
S_{t+1} = S_t - I_{t+1}, \quad
R_{t+1} = R_t + I_t,
\label{eqn:tsir_stoch_recur}
\end{gathered}
\end{equation}
where $\delta_0$ represents a point mass at 0; i.e 0 with probability 1. As before, we take the random initial state described in (\ref{eqn:vacc_init_cond}) to incorporate the vaccination parameter $v$.
Because all infected individuals recover at the end of each timestep, the number of total infections $Z$ can be written as
$Z = \sum_{t=0}^T I_t$. We write $Z(v,\beta,T)$ to emphasize the dependence upon the parameters $v$ and $\beta$.

%% file: gradients.tex
\vspace{-6pt}
\section{\emph{Gradient estimation}} \label{sec:grad_est}
We aim to estimate the derivative of $J(v,\beta,T) = \mathbb{E}[Z(v,\beta,T)]$
with respect to $v$ and $\beta$. 
Let $J_v(v,\beta,T)$, $J_\beta(v,\beta,T)$ denote the respective derivatives to be estimated, i.e.,
$$
J_v(v,\beta,T) = \frac{d \mathbb{E}[Z(v,\beta,T)]}{dv},
\quad
J_\beta(v,\beta,T) = \frac{d \mathbb{E}[Z(v,\beta,T)]}{d\beta}.
$$
Denote the (random) joint state at time $t$ by $\mathbf{X}_t := (S_t,I_t)$, and their realizations by $\mathbf{x}_t = (s_t,i_t)$. Then,
\begin{equation}
J(v,\beta,T) = 
\sum_{(\mathbf{x_0}, \ldots, \mathbf{x}_T) \in \mathcal{S}^T} 
\left(\sum_{t=0}^T i_t\right)
\mu(\mathbf{x}_0; v) 
\prod_{j=0}^{T-1} P(\mathbf{x}_{j+1} \mid \mathbf{x}_j, ; \beta)
\label{eqn:cost_fn}
\end{equation}
where $\mathcal{S}^T$ is the set of sequences of length $T$ valued in $\mathcal{S}$, $\mu(\cdot; v)$ is the distribution of the random initial state parameterized by $v$, and $P(\cdot \mid \cdot ; \beta)$ is the transition distribution parameterized by $\beta$. Note that (\ref{eqn:tsir_stoch_recur}) is written as a stochastic recursion, which gives a particular construction of the process, but does not directly specify $\mu(\cdot ; v)$ and $P(\cdot \mid \cdot ; \beta)$. Written explicitly, these distributions are
$$
\mu(s, i; v) = 
\begin{cases}
\binom{N - i_0}{N - i_0 -s}v^{N-i_0-s}(1-v)^{s} & 0 \leq s \leq N - i_0,i = i_0, \\
0 & \text{otherwise,}
\end{cases}
$$
$$
P(s', i' \mid s, i; \beta) = 
\begin{cases} 
   1 - \sum_{j=0}^{s-1} P_Y(j \mid s, i; \beta) & \text{if } s' = 0, \ i' = s, \text{ and } s, i > 0, \\
   P_Y(i' \mid s,i; \beta) & \text{if } s' = s - i', \ i' < s, \text{ and } s, i > 0, \\
   1 & \text{if } (i=0 \text{ or } s=0) \text{ and } s'=s, i'=0, \\
   0 & \text{otherwise,}
\end{cases}
$$
where $P_Y$ is the probability mass function of the negative-binomial distributed random variable $Y$ used in the construction of the TSIR process. $$
P_Y(y \mid s,i; \beta) = \frac{\Gamma(i + y)}{y! \Gamma(i)} \left(\frac{i}{i + \beta s i / N}\right)^i 
\left(
\frac{\beta si / N}{i + \beta s i /N}
\right)^y,
\quad
y = 0, 1, 2, \ldots
$$
This mass function corresponds to the form of the negative binomial with mean equal to $\lambda(s,i) = \beta s i / N$ and shape parameter equal to $i$. $P_Y(\cdot \mid s,i;\beta)$ is defined only if $s > 0$ and $i > 0$; and $\beta > 0$.

\vspace{-6pt}
\subsection{Weak-derivative gradient estimator for $J_v$} \label{sec:wd_grad}

Recall that the proportion immunized ($v$) enters the model via a random number of initially immune individuals, denoted $V \sim \text{Bin}(N-i_0, v)$ with PMF denoted by $\mu(\cdot;v)$. The weak-derivative (WD) method is based upon decomposing the derivative of $\mu$ into two PMFs. In particular, using Table 5.1 from \cite{fu2014stochastic},
$
\frac{d}{dv}\left[\mu(\mathbf{x}_0; v)\right]
=
(N - i_0)\left(\mu^+(\mathbf{x}_0;v) - \mu^-(\mathbf{x}_0; v)\right),
$
where $\mu^+$ is the distribution of $V^+ \sim 1 + \text{Bin}(N - i_0 -1,v)$ and $\mu^-$ is the distribution of $V^- \sim \text{Bin}(N-i_0-1,v)$. Define $Z^+(v,\beta,T)$ and $Z^-(v,\beta,T)$ as the outbreak sizes up to time $T$, given that the initially immune population is given by $V^+$ and $V^-$, respectively. The derivative $J_v$ can then be estimated by simulating versions of $Z^+$ and $Z^-$, taking the difference, and scaling by $N-i_0$.

Suppose we have $K$ independent realizations of tuples $(Z_k^+, Z_k^-)_{k=1}^K$, allowing dependence between $Z_k^+$ and $Z_k^-$. Then, we define the WD gradient estimator as
\begin{equation}
\widehat J_{v,K} = \frac{N-i_0}{K} \sum_{k=1}^K (Z^+_k - Z^-_k).
\label{eqn:wd_est}
\end{equation}
\begin{prop}
	The WD gradient estimator is unbiased: $J_v = \mathbb{E}[\widehat J_{v, K}]$.
	\label{prop:wd_unbiased}
\end{prop}

\noindent \textit{Proof.} We differentiate both sides of (\ref{eqn:cost_fn}) and interchange the derivative with summation over $\mathcal{S}^T$. This is permitted since $\mathcal{S}^T$ is a finite set. Since only $\mu(\cdot \mid ; v)$ depends on $v$,
$$
J_v = \frac{d}{dv}\mathbb{E}\left[Z(v,\beta, T)\right] = 
\sum_{\mathbf{x_0}, \ldots, \mathbf{x}_T \in \mathcal{S}^T} 
\left(\sum_{t=0}^T i_0\right)
\frac{d}{dv}\left[\mu(\mathbf{x}_0; v)\right]
\prod_{j=0}^{T-1} P(\mathbf{x}_{j+1} \mid \mathbf{x}_j, ; \beta).
$$

Plugging the decomposition of $\mu$ back into (\ref{eqn:cost_fn}) yields
\begin{align*}
J_v(v,\beta,T) &= 
(N- i_0)\sum_{\mathbf{x_0}, \ldots, \mathbf{x}_T \in \mathcal{S}^T} 
\left(\sum_{t=1}^T i_t\right)
\left[ \mu^+(\mathbf{x}_0; v) - \mu^-(\mathbf{x}_0; v)\right]
\prod_{j=0}^{T-1} P(\mathbf{x}_{j+1} \mid \mathbf{x}_j, ; \beta) \\
&= (N-i_0) \mathbb{E}[Z^+_k - Z^-_k]
\end{align*}
Averaging across $K$ independent replications does not change the mean. \qed

We consider a variance reduction approach for $\widehat J_v$ based on common random numbers (CRN) \cite[Chapter 11]{law2007simulation}. Intuitively, we want the difference in $Z^+$ and $Z^-$ to be completely due to the change in the parameter $v$, and not due to randomness. We induce dependence between $V^+$ and $V^-$ by noting that if we set $V^+ = 1 + V^-$,
each still has the correct marginal distribution. We simulate the process using the truncated stochastic recursion (\ref{eqn:tsir_stoch_recur}), and use the same random numbers for the dynamics, using an inverse-transform sampling procedure. Letting $F_Y(y \mid s, i; \beta)$ denote the CDF corresponding to the mass function $P_Y(y \mid s,i ; \beta)$, the next $Y$ is sampled using a common uniform random number $u_t$:
$Y^+_{t+1} = F^{-1}_{Y}(u_t \mid S_t^+, I_t^+; \beta ),\ 
Y^-_{t+1} = F^{-1}_{Y}(u_t \mid S_t^-, I_t^-; \beta).$

We denote the resulting CRN estimator as $\widehat{J}_{v,K}^{\text{CRN}}$. Intuitively, $V^+$ and $V^-$ under this construction should be negatively correlated since $V^+$ corresponds to one additional vaccination relative to $V^-$, reducing the outbreak size $Z$ by at least one. This indicates that this estimator can benefit from CRN. In Section \ref{sec:gradients_numerical}, we will compare the CRN construction to one where $Z^+$ and $Z^-$ are simulated independently.

Notice that this gradient estimator has an intuitive interpretation. We randomly designate one individual from the population. Then, we immunize a proportion $v$ of the remaining population. We then consider the difference in the outbreak size where the designated individual is immunized ($Z^+$), and where the designated individual is not ($Z^-$). Taking the difference of these two simulations and multiplying by the total population excluding the initially infected individuals yields the derivative, i.e., the marginal effect of vaccinating the designated individual.

\vspace{-6pt}
\subsection{Likelihood-ratio gradient estimator for $J_\beta$} \label{sec:LR_grad}

Next, we consider the likelihood-ratio (LR) gradient estimator for $J_\beta$. The LR method uses the fact that the derivative of the PMF may be represented as a product of the original PMF and the score function. In particular, notice that $\beta$ enters the model via the transition kernel $P(\cdot \mid \cdot; \beta)$, and
$$
\frac{d}{d\beta} P(\mathbf{x}_{j+1} \mid \mathbf{x}_j ; \beta) = 
\frac{\partial \log P(\mathbf{X}_{j+1}^{(k)} \mid \mathbf{X}_j^{(k)}; \beta)}{\partial \beta}
P(\mathbf{x}_{j+1} \mid \mathbf{x}_j ;\beta),
$$
for $P(\mathbf{x}_{j+1} \mid \mathbf{x}_j ; \beta) > 0$. This allows us to write the derivative as an expectation with respect to the distribution of the epidemic process, i.e., the gradient estimator can be calculated using sample paths of the original simulation. Applying this along with the product rule to ${\prod_j P(\mathbf{x}_{j+1} \mid \mathbf{x}_j ; \beta)}$ yields a summation of score function terms, multiplied by the PMF: 
\begin{equation}
\frac{d}{d\beta} \left[\prod_{j=1}^{T-1}P(\mathbf{x}_{j+1} \mid \mathbf{x}_j ; \beta)\right] = 
\left(\sum_{j=0}^{T-1}\frac{\partial \log P(\mathbf{X}_{j+1}^{(k)} \mid \mathbf{X}_j^{(k)}; \beta)}{\partial \beta}\right)
\prod_{j=1}^{T-1}P(\mathbf{x}_{j+1} \mid \mathbf{x}_j ;\beta).
\label{eqn:lr_trick}
\end{equation}
Take $K$ i.i.d samples of the epidemic process, i.e., let $\left(\mathbf{X}_t^{(k)}\right)_{t=0}^T = (S_t^{(k)}, I_t^{(k)})_{t=0}^T$ be the $k$th sample path of the process up to time $T$. The LR gradient estimator is
\begin{equation}
	\widehat J_{\beta, K} = 
	\frac{1}{K}
		\sum_{k=1}^K
		Z_k\left(
		\sum_{j=0}^{T-1} 
		\frac{\partial \log P(\mathbf{X}_{j+1}^{(k)} \mid \mathbf{X}_j^{(k)}; \beta)}{\partial \beta}\right).
	\label{eqn:LR_est}
	\end{equation}
\begin{prop}
	The LR gradient estimator is unbiased: $J_\beta(v,\beta,T) = \mathbb{E}[\widehat J_{\beta, K}]$.
	\label{prop:lr_unbiased}
\end{prop}

\noindent \textit{Proof.} Differentiating both sides of (\ref{eqn:cost_fn}) and applying (\ref{eqn:lr_trick}),
\begin{align*}
J_\beta = \frac{d}{d\beta}\mathbb{E}\left[Z(v,\beta, T)\right] 
&= 
\sum_{\mathbf{x_0}, \ldots, \mathbf{x}_T \in \mathcal{S}^T} 
\left(\sum_{t=0}^T i_t\right)
\mu(\mathbf{X}_0; v) 
\frac{d}{d \beta}\left[\prod_{j=0}^{T-1} P(\mathbf{x}_{j+1} \mid \mathbf{x}_j, ; \beta)\right] \\
&=
\mathbb{E}\left[
	\left(\sum_{t=0}^T I_t\right)
	\left(\sum_{t=0}^{T-1} \frac{\partial \log P(\mathbf{X}_{j+1} \mid \mathbf{X}_j ; \beta)}{\partial \beta}\right)
\right].
\end{align*}
By directly taking the derivative, 
\begin{equation}
	\label{eqn:score_function}
\frac{\partial \log P(s', i' \mid s, i)}{\partial \beta} = 
	\begin{cases} 
	\dfrac{- \frac{si}{N} \displaystyle\sum_{j=0}^{s-1} \left[ \left( \frac{j}{\beta s i /N} - \frac{j + i}{i + \beta s i /N} \right) P_Y(j \mid s, i) \right]}{1 - \displaystyle\sum_{j=0}^{s-1} P_Y(j \mid s, i)} 
	& \substack{s' = 0\\ i' = s\\ s,i > 0,} \\[3em]
	\dfrac{si}{N} \left( \dfrac{i'}{\beta s i /N} - \dfrac{i' + i}{i + \beta s i /N} \right) 
	& \substack{s' = s - i'\\ i' < s \\ s, i > 0,}\\[1.5em]
	0 & \text{otherwise.}
	\end{cases}
\end{equation}

Averaging over $K$ independent simulations does not change the mean.
\qed

For variance reduction, we first simplify the estimator by verifying the following equality.
\begin{lemma}
	\label{prop:sum_simplify}
$
\mathbb{E}\left[
	\left(\sum_{t=0}^T I_t\right)
	\left(\sum_{t=0}^{T-1} \frac{\partial \log P(\mathbf{X}_{j+1} \mid \mathbf{X}_j ; \beta)}{\partial \beta}\right)
\right]
=
\mathbb{E}\left[
\sum_{j=0}^{T-1} \sum_{t = j+1}^T I_t 
\frac{\partial \log P(\mathbf{X}_{j+1} \mid \mathbf{X}_j; \beta)}{\partial \beta}
\right].
$
\end{lemma}
\noindent \textit{Proof.} The simplification follows if we show that for $j \geq t$,
$
\mathbb{E}\left[I_t \frac{\partial \log P(\mathbf{X}_{j+1} \mid \mathbf{X}_j; \beta)}{\partial \beta}\right] = 0.
$
Let $\mathcal{H}_j = \sigma(\mathbf{X}_\ell : 0 \leq \ell \leq j)$. Applying iterated expectations yields
$
\mathbb{E}\left[I_t
    \mathbb{E}\left[        
    \frac{\partial \log P(\mathbf{X}_{j+1} \mid \mathbf{X}_j)}{\partial \beta}\mid \mathcal{H}_j
    \right]
\right],
$
since $I_t$ is $\mathcal{H}_j$-measurable. However, for any fixed $\mathbf{X}_j$,
$
 \mathbb{E}\left[        
    \frac{\partial \log P(\mathbf{X}_{j+1} \mid \mathbf{X}_j)}{\partial \beta} \mid \mathbf{X}_j
\right] = 0,
$
since this is equivalent to
$
\sum_{\mathbf{X}_{j+1}}\frac{\partial}{\partial \beta} P(\mathbf{X}_{j+1} \mid \mathbf{X}_j; \beta) = \frac{\partial}{\partial \beta} \sum_{\mathbf{X}_{j+1}} P(\mathbf{X}_{j+1} \mid \mathbf{X}_j ; \beta) = \frac{\partial}{\partial \beta} \left[1\right] = 0.
$ \qed

Fixing a time $t$, denote the sample mean of $I_t$ across these $K$ replications be $\overline{I_t}= \frac{1}{K}\sum_{k=1}^K I_t^{(k)}$. We use $\overline{I}_t^{(k)}$ as a baseline in order to center the samples of $I_t$ and further reduce the variance. Overall, the LR gradient estimator with variance reduction is 
\begin{equation}
	\widehat J_{\beta, K}^{\text{VR}} = 
\frac{1}{K}\sum_{k=1}^K\sum_{j=0}^{T-1} \sum_{t = j+1}^T \left(I_t^{(k)} - \overline{I_t}\right) 
\frac{\partial \log P(\mathbf{X}_{j+1}^{(k)} \mid \mathbf{X}_j^{(k)}; \beta)}{\partial \beta}.
\label{eqn:LR_est_VR}
\end{equation}

The usage of $\overline I_t$ as a baseline introduces a bias of order $O(1/K)$: 
\begin{prop}
$
	\mathbb{E}[\widehat{J}_{\beta, K}^{\text{VR}}]
=
	\left(\frac{K-1}{K}\right) \frac{d \mathbb{E}[Z(v,\beta,T)]}{d\beta}.
$
\end{prop}
\noindent \textit{Proof.} 
Let $S_{j,k}(\beta) = \frac{\partial \log P(\mathbf{X}_{j+1}^{(k)} \mid \mathbf{X}_j^{(k)}; \beta)}{\partial \beta}$.
Taking the expectation inside the sum and using Lemma \ref{prop:sum_simplify},
{
$
\mathbb{E}[\widehat{J}_{\beta,K}^{\text{VR}}]
=
\frac{d \mathbb{E}[Z(v,\beta, T)]}{d \beta} -
\sum_{j=0}^{T-1} \sum_{t=j+1}^T 
\mathbb{E}\left[
\overline{I_t}  
	S_{j,k}(\beta)
	\right].
$
}
The second term is the bias term. By the definition of $\overline{I_t}$,
$
\mathbb{E}\left[
\overline{I_t}  
S_{j,k}(\beta)
\right] = \frac{1}{K}\sum_{\ell=1}^K 
\mathbb{E}\left[
I_t^{(\ell)}   
S_{j,k}(\beta)
\right].
$
Because of independence, if $\ell \neq k$,
$
\mathbb{E}\left[
I_t^{(\ell)}   
S_{j,k}(\beta)
\right] =\mathbb{E}[I_t^{(\ell)}]
\mathbb{E}\left[
S_{j,k}(\beta)
\right] = 0,
$ meaning the only surviving term is $\ell = k$.
Thus, the bias term becomes
$
\frac{1}{K}
\sum_{j=0}^{T-1} \sum_{t=j+1}^T 
\mathbb{E}\left[
I_t^{(k)}
S_{j,k}(\beta)
\right] = \frac{1}{K} \frac{d \mathbb{E}[Z(v,\beta,T)]}{d\beta}.
$ Plugging this back into the equation for $\mathbb{E}[\widehat{J}_{\beta, K}^{\text{VR}}]$ yields the desired equality.
\qed

%% file: parameter_uncert.tex
\vspace{-6pt}
\section{
\emph{Parametric uncertainty}} \label{sec:parameter_uncertainty}

Assuming we are given a probability distribution for $v$ and $\beta$ representing the uncertainty from parameter estimation, we wish to incorporate the uncertainty in $v$ and $\beta$ into both sensitivity analysis and optimization. Examples of fitting SIR models to data in a Bayesian context can be found in \cite{capistran2012towards}. The following discussion of optimization under parameter uncertainty can also be contextualized in light of the Bayesian risk optimization framework (\cite{wu2018bayesian}); in particular, a gradient-based approach was proposed in \cite{cakmak2021solving}.

First, consider the context where both parameters are uncertain, focusing on the application to sensitivity analysis. Let $(\hat v,\hat \beta)$ be the (random) parameters satisfying $\mathbb{E}[\hat v] =v_0$, $\mathbb{E}[\hat \beta] = \beta_0$. We are interested in the \textit{expected sensitivities}
$
\xi_v(v_0,\beta_0,T) = \mathbb{E}_{\hat v, \hat \beta} \left[J_v(\hat v, \hat \beta, T)\right]$,
$
\xi_\beta(v_0,\beta_0,T) = \mathbb{E}_{\hat v, \hat \beta} \left[J_\beta(\hat v, \hat \beta, T)\right],
$
which are the derivatives averaged across the uncertainty in the parameters. One can interpret the expected sensitivity as the expected effect of a small change in $v$ or $\beta$ averaged across all plausible parameter values. 

Next, we consider the optimization context. Here, one of the parameters $v$ and $\beta$ is random, and the other variable is a decision variable (and thus is known). The \textit{derivatives under parameter uncertainty} are
$
J_v(\hat \beta) = \frac{d}{dv} \mathbb{E}_{\hat \beta}\left[ \mathbb{E}[Z(v,\hat \beta, T)]\right],
$
$
J_\beta(\hat v) = 
\frac{d}{d\beta} \mathbb{E}_{\hat v}\left[ \mathbb{E}[Z(\hat v, \beta, T)]\right].
$
To show that the gradient estimators $\widehat J_v$ and $\widehat J_\beta$ remain unbiased under parameter uncertainty, we need to show the interchange of derivative through the outer expectations \cite[Section 5.3.5]{fu2014stochastic}. 
In Proposition \ref{prop:param_uncert}, we verify that the expected sensitivities are finite, and that the gradient estimators under parameter uncertainty are unbiased. The following lemma is used in Proposition \ref{prop:param_uncert} and bounds the score function (\ref{eqn:score_function}). 

\begin{lemma}
	\label{lemma:score_bound}
$
\left|\frac{\partial \log P(\mathbf{X}_{j+1} \mid \mathbf{X}_j; \beta)}{\partial \beta}\right|
\leq
\frac{N}{\beta}.
$
\end{lemma}
\noindent \textit{Proof.} In the second case of (\ref{eqn:score_function}), we obtain the bound
$
\frac{si}{N} \left(
	\frac{i'}{\beta si /N} -
	\frac{i' + i}{i + \beta si /N}
\right)
\leq
\frac{si}{N} \frac{i'}{\beta s i /N}
\leq N/\beta,
$
using that $0 < s,i,i'\leq N$ with probability one. The expression in the first case of (\ref{eqn:score_function}) is equivalent to 
$
\frac{\partial}{\partial \beta} \log \mathbb{P}(Y \geq s ; \beta, s, i).
$
Here, we use the parameterization of the negative binomial where $r$ is the number of successes before the experiment is stopped, and $p$ is the probability of success.
Then $Y$ is negative-binomial with $r=i$ and $p=i/(i + \beta si / N)$. Define
$
q = 1-p = \frac{\beta s}{N + \beta s},
$
and recalling that the CDF of a negative-binomial is given by the regularized incomplete beta function $I_p$,
$
\mathbb{P}(Y \geq s ; \beta, s, i) = I_q(s,i) = {\frac{1}{B(s,i)} \int_0^{q} t^{s-1} (1-t)^{i-1} \text{ d}t},
$
where $B(\cdot,\cdot)$ is the beta function.
By the chain rule,
\begin{align*}
\frac{\partial}{\partial \beta} \log \mathbb{P}(Y \geq s ; \beta, s, i)
&=
\left[I_q(s,i)\right]^{-1} \frac{\partial}{\partial \beta}[I_q(s,i)] \frac{\partial q}{\partial \beta} \\
&=
\left[
	\int_0^q t^{s-1}(1-t)^{i-1} \text{ d}t
\right]^{-1}
\left[
	q^{s-1}(1-q)^{i-1}
\right]
\left[
	\frac{s}{N + \beta s} - \frac{\beta s^2}{(N+\beta s)^2}
\right] \\
&\leq
\left[ 
	\frac{s}{q^s} \frac{1}{(1-q)^{i-1}}
\right]
\left[
	q^{s-1}(1-q)^{i-1}
\right]
\left[
	\frac{s}{N + \beta s} - \frac{\beta s^2}{(N+\beta s)^2}
\right] \\
&= \frac{s}{q}
\left[
	\frac{s}{N + \beta s} - \frac{\beta s^2}{(N+\beta s)^2}
\right]
=
\frac{1}{\beta} \left[s - \frac{\beta s^2}{N + \beta s}\right] \leq N/\beta.
\qed
\end{align*}

\begin{prop}
    \label{prop:param_uncert}
    Assume that $\hat v$ is a random variable with support in $(0,1)$, and that $\hat \beta > 0$ with $\mathbb{E}[1 /\hat \beta] < \infty$. Then,
	\textbf{(i)} $\mathbb{E}_{\hat v, \hat \beta} \left[\frac{d}{dv}\mathbb{E}[Z(\hat v,\hat \beta,T)]\right] < \infty,$
	\textbf{(ii)} $\mathbb{E}_{\hat v, \hat \beta} \left[\frac{d}{d\beta}\mathbb{E}[Z(\hat v,\hat \beta,T)]\right] < \infty \label{eqn:beta_uncert_finite},$
	\textbf{(iii)} $\mathbb{E}_{\hat \beta} [\widehat{J}_{v}( v, \hat \beta, T)] 
	    = J_v(\hat \beta) \label{eqn:v_uncert_unbiased},$
	\textbf{(iv)}  $\mathbb{E}_{\hat v} [\widehat{J}_{\beta}(\hat v, \beta, T)] 
	    = J_\beta(\hat v)$.
\end{prop}

\noindent\textit{Proof.}
First, we analyze the derivative with respect to $v$. For all $v \in (0,1)$, $\beta > 1$,
$\left|\frac{d}{dv}\mathbb{E}[Z]\right| = |(N-i_0) \mathbb{E}[Z^+ - Z^-]| \leq 2N^2$
since $Z^+ \leq N$, $Z^- \leq N$. Thus, $\frac{d}{dv}\mathbb{E}[Z(v,\beta, T)]$ is a bounded function, implying \textbf{(i)}. By the mean value theorem, observe that for all $\beta$,
$$
\frac{\mathbb{E}[Z(v+ \varepsilon, \beta, T) - Z(v, \beta, T)]}{\varepsilon}
\leq \sup_{\tilde v \in [v, v+\varepsilon]} \frac{d}{dv} \mathbb{E}[Z(\tilde v,\beta, T)] \leq 2N^2,
$$
enabling the derivative to pass through the expectation with respect to $\hat \beta$ by the dominated convergence theorem. By Proposition \ref{prop:wd_unbiased}, we can exchange derivative with the inner expectation to obtain \textbf{(iii)}.

Next, we analyze the derivative with respect to $\beta$. Using (\ref{eqn:LR_est}) and $Z \leq N$,
\begin{equation}
\left|\frac{d}{d\beta} \mathbb{E}[Z(v,\beta, T)]\right| \leq
N \sum_{t=0}^{T-1} \mathbb{E}\left|\frac{\partial \log P(\mathbf{X}_{j+1} \mid \mathbf{X}_j; \beta)}{\partial \beta}\right|.
\label{eqn:beta_bound}
\end{equation}

Using Lemma \ref{lemma:score_bound} and substituting $\hat \beta$ for $\beta$, the expectation of the right-hand side of (\ref{eqn:beta_bound}) is finite by assumption, implying \textbf{(ii)}. By a similar application of the dominated convergence theorem, the interchange over the outer expectation with respect to $\hat v$ is permissible. Using Proposition \ref{prop:lr_unbiased}, we exchange with the inner expectation to obtain \textbf{(iv)}.\qed

This proposition justifies a two-stage sampling procedure to estimate the expected sensitivities and derivatives under parameter uncertainty. Parameter values are first drawn from their distributions, and epidemic trajectories are then simulated conditionally on these parameters. In this work, we adopt a single-level Monte Carlo scheme. For each sampled parameter pair $(\hat v^{(k)}, \hat \beta^{(k)})$, we generate a single replication of the gradient estimator. Thus, drawing $K$ samples from the distribution of $(\hat v, \hat \beta)$ yields $K$ realizations of the estimator, using the variance-reduced estimators $\widehat J^{\text{VR}}_{\beta,K} $ and $\widehat J^{\text{CRN}}_{v,K}$ in Sections \ref{sec:wd_grad} and \ref{sec:LR_grad}.

%% file: numerical_sec_intro.tex
\vspace{-6pt}
\section{Numerical experiments} \label{sec:numerical}

In Section \ref{sec:gradients_numerical}, we compare the proposed estimators against the following one-sided finite-difference (FD) estimators:
$$
\widehat J_{v, \varepsilon_v}^{FD}
=
\frac{1}{K}\sum_{k=1}^K
\frac{Z_k(v+\varepsilon_v, \beta, T)-Z_k(v,\beta,T)}{\varepsilon_v},
\quad
\widehat J_{\beta, \varepsilon_\beta}^{FD}
=
\frac{1}{K}\sum_{k=1}^K
\frac{Z_k(v,\beta+\varepsilon_\beta,T)-Z_k(v,\beta,T)}{\varepsilon_\beta},
$$
and final size equations (FSE). In Section \ref{sec:sensitivity_analysis}, we use the proposed estimators to investigate how the sensitivities vary over the parameter space and how parameter uncertainty changes the sensitivity landscape. We interpret the estimated sensitivities in light of phenomena like herd immunity, and explain how the sensitivity can be used to estimate the marginal impact of interventions like vaccination or contact reduction.
In Section \ref{sec:optimization}, we illustrate an application of the proposed estimators through two stylized optimization problems, which investigate tradeoffs between the costs of infections and the costs of interventions from the viewpoint of a social planner. These examples also aim to demonstrate how the estimators can be used to examine the qualitative effects of parameter uncertainty on resulting policies.

%% file: sensitivity_results.tex
\vspace{-6pt}
\subsection{\emph{Comparison to finite-differences \& final size equations}} 
\label{sec:gradients_numerical}

Table~\ref{tab:grad_est_comparison} reports the results for the gradient estimators with and without variance reduction, the finite-difference estimators, and the final size equations, for four sets of values of $v$ and $\beta$, averaged over 10{,}000 independent replications, with $N=1000$, $T=10$, and $i_0=1$ for all simulations, where FSE is calculated using the approach described in Section \ref{sec:final_size_equations}; WD-IND is WD with the pairs $Z^+_k$ and $Z^-_k$ simulated independently; WD-CRN is WD with $Z^+_k$ and $Z^-_k$ simulated using CRN, as described at the end of Section \ref{sec:wd_grad}; LR uses (\ref{eqn:LR_est}); LR-VR is LR with variance reduction, using (\ref{eqn:LR_est_VR}); FD-CRN means finite-difference with CRN, using the same construction as WD-CRN, and
using perturbation sizes of $\varepsilon_v = 0.005$ and $\varepsilon_\beta = 0.05$.
\begin{table}[h]
\centering
\label{tab:grad_est_comparison}

\sisetup{
    round-mode = figures,
    round-precision = 3,
    scientific-notation = true,
    output-exponent-marker = \mathrm{e},
}

\begin{tabular}{l S @{\ } l S @{\ } l S @{\ } l S @{\ } l}
\multicolumn{9}{l}{\textbf{Estimates of $J_v$}} \\
	Estimator & \multicolumn{2}{c}{$v=0.1, \beta=2$} & \multicolumn{2}{c}{$v=0.6, \beta=2$} & \multicolumn{2}{c}{$v=0.1, \beta=4$} & \multicolumn{2}{c}{$v=0.6, \beta=4$} \\
\midrule
	FSE    & -1533.30 &           & 0.00    &           & -1171.76 &           & -1517.49 &                 \\
	WD-IND & -34.5654   & (\num{3894.75}) & 60.6393   & (\num{119.76})  & -1997.5005 & (\num{5672.786}) & -196.1037  & (\num{1544.656}) \\
	WD-CRN & -1283.515  & (\num{114.138})   & -23.976  & (\num{3.3794})    & -990.9081  & (\num{155.4935})   & -862.9362  & (\num{45.4427})   \\
	FD-CRN & -1153.74 & (\num{38.57687})   & -27.24  & (\num{2.38528})    & -980.34  & (\num{67.55363})   & -889.26  & (\num{25.217})   \\
\midrule
\multicolumn{9}{l}{\textbf{Estimates of $J_\beta$}} \\
Estimator & \multicolumn{2}{c}{$v=0.1, \beta=2$} & \multicolumn{2}{c}{$v=0.6, \beta=2$} & \multicolumn{2}{c}{$v=0.1, \beta=4$} & \multicolumn{2}{c}{$v=0.6, \beta=4$} \\
\midrule
	FSE    & 448.58  &           & 0.00    &           & 52.16    &           & 110.53   &                 \\
	LR     & 446.624  & (\num{29.857})   & 5.264    & (\num{0.3689})    & 14.757    & (\num{25.005})   & 76.737    & (\num{3.885})    \\
	LR-VR  & 397.7314  & (\num{13.247})   & 5.31367    & (\num{0.2837})    & 61.38626    & (\num{8.67})    & 69.897    & (\num{2.06385})    \\
	FD-CRN & 451.054  & (\num{9.165})   & 5.37    & (\num{0.32656})    & 57.122    & (\num{9.96198})   & 73.53    & (\num{2.235})    \\
\bottomrule
\end{tabular}
	\caption{\textbf{Comparison of gradient estimators for selected values of $v$ and $\beta$, based on 10,000 independent replications.} Standard errors are shown in parentheses. The FD-CRN estimators used perturbation sizes of $\varepsilon_v=0.005$ and $\varepsilon_\beta=0.05$. We set $N=1000, T=10, i_0=1$ for all simulations. }

\label{tab:grad_est_comparison}
\end{table}

\textbf{Estimates of $J_v$}: The sign is primarily negative. The magnitude of most of the estimates appear to be on the order of $N=1000$, the size of the population, except for when $v=0.6$ and $\beta=2$. The WD-IND estimates have very large standard errors and sometimes give positive derivatives. Using CRN with the WD estimator substantially reduces the variance, but the WD-CRN estimates still have larger standard errors than the FD estimator. The WD-CRN and FD-CRN estimates agree reasonably well, considering the standard errors, but differ dramatically from the FSE and WD-IND estimates. Generally, the FSE estimates appear to have larger magnitudes, except when they vanish for $v=0.6$ and $\beta=2$. 

\textbf{Estimates of $J_\beta$}: The sign is primarily positive. The magnitude of the estimates is smaller compared to the estimates of $J_v$, and is at most half the population size; it is substantially smaller when $v=0.6$ and $\beta=2$. The LR estimates have more reasonable standard errors, compared to WD-IND. However, we still see a benefit when using LR-VR, as the variance is sometimes less than half of the variance of LR, resulting in similar performance to the FD-CRN estimator. The LR, LR-VR, and FD-CRN all agree reasonably well, considering the standard errors, but not with the FSE. The FSE estimates again generally appear to be larger, except when it vanishes again for $v=0.6$ and $\beta=2$. 

\textbf{Discussion.} Intuitively, $J_v$ should be negative, as increasing vaccinations decreases the number of infections. Likewise, $J_\beta$ should be positive, as increasing contacts increases the final size. This monotonicity with respect to $\beta$ is shown in Chapter 3 of \cite{Andersson_Britton_2000}. 

Because $\mathcal{R}_{\text{eff}} < 1$ when $v=0.6$ and $\beta=2$, no large outbreak can occur, and so the magnitude of $J_v$ and $J_\beta$ is smaller in this case. The FSE neglects this case due to the assumption that $\mathcal{R}_{\text{eff}} > 1$, yielding a derivative of 0, which explains the discrepancy with the gradient estimators. In the other cases, the final size equation may give a larger derivative due to the effect of the time horizon. Individuals may become infected past the finite time horizon of $T$, which is taken into account by the final size equations, but not by the gradient estimators. 

In summary, while the WD-CRN estimator for $J_v$ had higher variance in the simulation experiments, it is unbiased and has an attractive theoretical interpretation, whereas the FD-CRN estimator for $J_v$ had lower variance but is biased and requires selecting an appropriate perturbation size. Both have similar computational cost, as both require two simulations.
For estimating $J_\beta$, the LR-VR estimator is superior to the FD-CRN estimator, as it only requires one simulation run per replication, has similar variance, and is unbiased, or in the case of the variance-reduced estimator, asymptotically unbiased.

\vspace{-6pt}
\subsection{
\emph{Sensitivity analysis under parametric uncertainty}}\label{sec:sensitivity_analysis}
We estimate the derivatives according to the FSE, the WD-CRN, and the LR-VR, again taking $N = 1000$, $i_0=1$, and $T=10$. We consider a grid of values with $\beta$ ranging between 2 and 4, and $v$ ranging between 0.01 and 0.99, using $K=10{,}000$ independent samples for the gradient estimators without parameter uncertainty. The averaged estimates are shown in Figure \ref{fig:contour_plots}.

We also estimate the expected sensitivity to assess the effect of parameter uncertainty. Uncertainty in $v$ is modeled using the Beta distribution, which is conjugate to the Binomial model for proportions. We set the shape parameters to be $p = C v_0$ and $q = C(1 - v_0)$, where $C > 0$ controls the concentration of the distribution. Under this parameterization, the mean is fixed at $v_0$, and increasing $C$ reduces the uncertainty by reducing the variance.
To model uncertainty in $\beta$, we use a Gamma distribution, which is conjugate to rate parameters in Poisson-type models. Let $a$ denote the shape parameter and $\lambda$ the rate parameter. We set $\lambda = a / \beta_0$, so that the mean is fixed at $\beta_0$. Under this parameterization, the variance decreases as $a$ increases, and the integrability condition required in Proposition \ref{prop:param_uncert} is satisfied provided that $a > 1$. Fixing $C = 10$ and $a = 10$, we let $v_0$ range between 0.01 and 0.99, and $\beta_0$ range between 2 and 4. The parameters are sampled independently, and we take $K=10,000$ independent replications as before. The estimated expected sensitivities are shown in Figure \ref{fig:contour_plots}.

\begin{figure}[h!]
	\centering
	\includegraphics[width=0.7\linewidth]{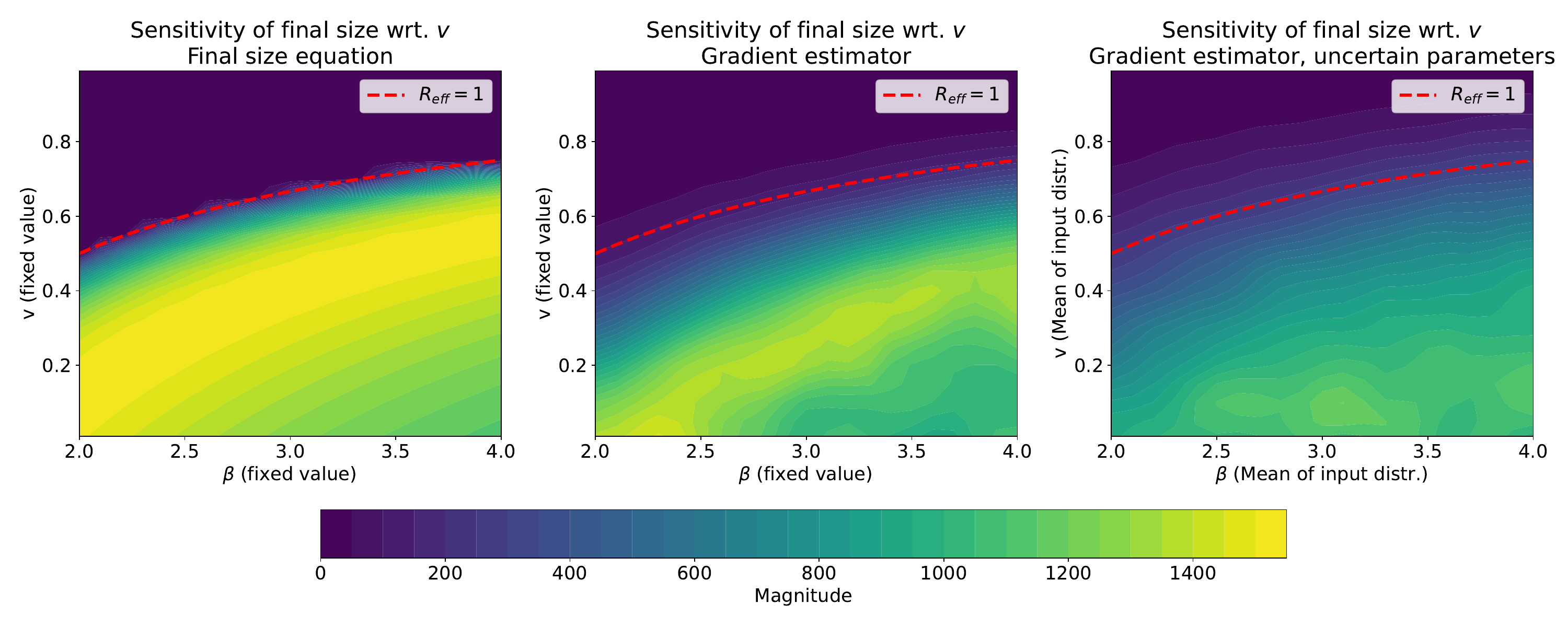}
	\includegraphics[width=0.7\linewidth]{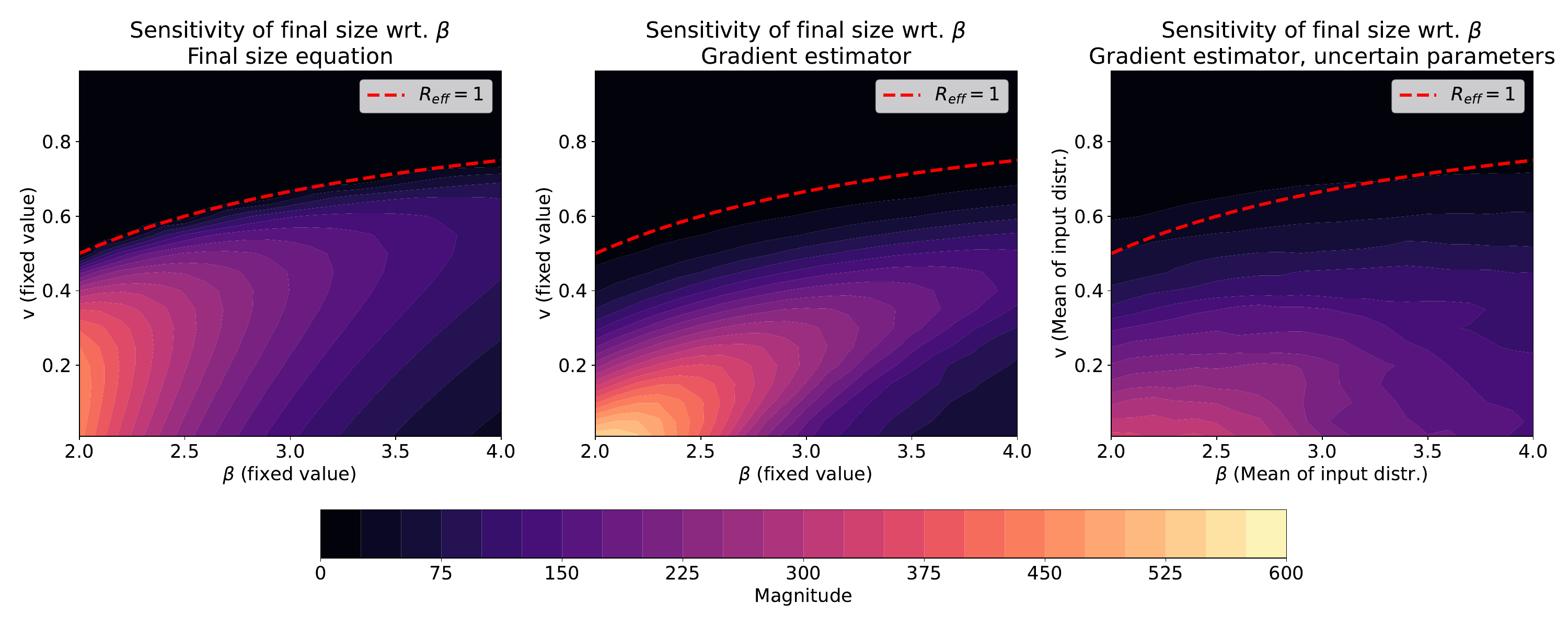}
	\caption{
		\textbf{Contour plots of the derivative-based sensitivity from different methods vs.\ the parameters $\beta$ and $v$.}
		The top row shows the estimated sensitivities of the total infections with respect to $v$, and the bottom row shows the estimated sensitivities of the total infections with respect to $\beta$. From left to right, we show the sensitivities from the final size equations; the derivatives $J_v$ and $J_\beta$; and the expected sensitivities $\xi_v$ and $\xi_\beta$. The gradient estimators $\widehat J_{v,K}^{\text{CRN}}$ and $\widehat J_{\beta,K}^{\text{VR}}$ are used to estimate the derivatives and expected sensitivities, with 10,000 independent replications used to estimate $J_v$, $J_\beta$, $\xi_v$, and $\xi_\beta$. The x- and y-axes are the values of $\beta$ and $v$, respectively; for the expected sensitivity the x- and y-axes are the means of the input distributions $\beta_0$ and $v_0$.
	}
	\label{fig:contour_plots}
\end{figure}

\textbf{Sensitivity with respect to $v$}: The sensitivity is small in all cases when $\mathcal{R}_\text{eff} < 1$. On the other hand, when $\mathcal{R}_\text{eff}>1$, the sensitivities are on order of the population size, $N=1000$. We observe that the final size equations yield the largest sensitivities, then the gradient estimator, then the gradient estimator under uncertainty. However, under parametric uncertainty, some contours stretch further into the region $\mathcal{R}_{\text{eff}} < 1$ compared to the final size equations or the gradient estimator. The result of the final size equation and the gradient estimator exhibit a peaked region when $\mathcal{R}_{\text{eff}} > 1$. However, this does not happen for the gradient estimator under parametric uncertainty.

\textbf{Sensitivity with respect to $\beta$:} The contours from the final size equation and the gradient estimator exhibit similar parabolic shapes and stretch out as $\beta$ increases. For each fixed $v$, the sensitivity as a function of $\beta$ has a single interior maximum; as $v$ increases, this maximizer shifts rightward and the peak flattens. The sensitivity is highest when $\beta$ and $v$ are both relatively low, and is zero when $\mathcal{R}_{\text{eff}} < 1$. The contours from the final size equation cover a larger region of the parameter space, whereas those from the gradient estimator are more concentrated near the center of the plot. Under parametric uncertainty, the profile of the contours becomes much flatter and no longer has a parabolic shape. For a fixed $v$, the sensitivity as a function of $\beta$ appears to decrease nearly monotonically, or only has a very small peak. In this case, some contours stretch beyond the boundary $\mathcal{R}_{\text{eff}} = 1$.

\textbf{Discussion}: The derivative with respect to $v$ can be interpreted via the indirect effects of vaccination. In addition to the direct effect of an individual being vaccinated (i.e., avoiding infection), vaccination prevents that individual from infecting others, enabling \textit{herd immunity} (\cite{metcalf2015understanding}). When $\beta$ is larger and $v$ is smaller (i.e., the bottom right corner of the plots), the derivative is around 1000. This corresponds to the change in infections if $v$ were increased by 1, i.e., if the entire population were vaccinated. Dividing the derivative by the size of the population yields approximately one infection prevented per vaccination, indicating no indirect benefits. On the other hand, if $v$ is around 0.2 and $\beta$ is near 3.0, the magnitude of the derivative is around 1400. The same scaling suggests that one vaccination prevents approximately 1.4--1.5 infections. \cite{Lin_Hamedmoghadam_Shorten_Stone_2024} investigated the direct and indirect effects of vaccination using final size equations. Under their framework, the indirect effects of vaccination as a function of $v$ exhibits a peak, similar to our results. They also found that increasing $\mathcal{R}_{\text{eff}}$ (which is equivalent to increasing $\beta$) shifts the peak toward larger values of $v$. Both observations are consistent with our numerical results.

However, under parametric uncertainty, these indirect effects become less apparent. Using the same scaling, incorporating parameter uncertainty reduces the net effect to approximately 0.8--1.2 infections prevented per vaccination. This can be interpreted as averaging the derivative with respect to the parameter distribution, which flattens out the sensitivity landscape. Practically, this suggests that if there is large uncertainty in the parameters, a planner should be wary about relying on indirect effects. The average effect of a vaccination over all plausible scenarios (as captured by the distributions of $\hat v$, $\hat \beta$) may not be as large as expected if the parameters were known. 

The behavior of the sensitivity with respect to $\beta$ can be interpreted as saturation. Because the population is finite, once a large fraction of individuals has been infected, further increases in $\beta$ have a decreasing marginal effect on the total number of infections. Thus, for fixed $v$, the sensitivity initially increases and then decreases as $\beta$ becomes large. $\beta$ saturates quickly when $v$ is small, and more slowly when $v$ is large. Additionally, when $v$ is small, the sensitivity peaks at a larger value, since the total pool of susceptibles is larger; when $v$ is large, the sensitivity peaks at a smaller value.

A similar marginal interpretation can be obtained by noting that the total number of contacts over the simulation horizon is approximately $N \beta T$. Thus, a unit decrease in $\beta$ corresponds to a reduction of $NT$ contacts. For $N=1000$ and $T=10$, this equals $10^4$ contacts; comparing this to peak derivative values in Figure \ref{fig:contour_plots} of about $600$ suggests a marginal reduction of roughly $0.06$ infections per contact in the most favorable regime. Under parametric uncertainty, this peak effect is roughly halved, indicating that the marginal effectiveness of contact-reducing interventions is also diminished when averaged across parameter uncertainties.

In effect, parameter uncertainty flattens the sensitivity profile across the parameter space. When the true parameter values are uncertain, it becomes more difficult to determine whether mechanisms such as herd immunity or saturation dominate. In terms of potential policy implications, these results suggest that stronger interventions may be required to achieve similar expected outcomes when incorporating parameter uncertainty. This helps explain the more conservative solutions obtained in the optimization analysis in Section~\ref{sec:optimization}. 


%% file: opt_results.tex
\subsection{
\emph{Optimization problems}} \label{sec:optimization}

First, consider a problem in which the decision variable is the proportion immunized ($v$). This problem, which we call the \textit{vaccination problem}, involves selecting an appropriate immunization level to balance the costs of immunization and infection. Assuming $\beta$ is uncertain and represented by the random variable $\hat \beta$,

\textbf{Problem 1 \textit{(Vaccination problem.)}}
\begin{equation}
	\quad
	\min_{v \in \mathcal{V}}\mathbb{E}[Z(v,\hat \beta,T)] + c_v\left(\frac{v}{1-v}\right) 
	\text{ s.t } 0 < v < 1.
	\label{eqn:opt_problem_1}
\end{equation}
where $c_v$ is a vaccination cost scaling parameter, and $\mathcal{V} \subset [0,1]$ is the feasible region. The cost term $c_v \, v/(1-v)$ is chosen to capture increasing marginal costs of vaccination as coverage approaches one. This reflects the intuition that achieving high levels of coverage typically requires greater effort, for example due to outreach or logistical challenges in reaching hesitant or hard-to-access populations. 

Next, consider a complementary problem in which the decision variable is the contact rate ($\beta$), which we call the \textit{contact problem}. This captures interventions such as social distancing or activity restrictions, which reduce transmission at the cost of limiting social and economic interactions. Assuming that $v$ is uncertain and represented by the random variable $\hat v$,

\textbf{Problem 2 \textit{(Contact problem.)}}
\begin{equation}
	\quad
	\min_{\beta \in \mathcal{B}} \mathbb{E}[Z(\hat v,\beta,T)] + c_\beta \left(\frac{1}{\beta - t_\beta}\right) 
	\text{ s.t } t_\beta < \beta,
	\label{eqn:opt_problem_2}
\end{equation}
where $t_\beta$ is a lower bound on the contact rate, $c_\beta$ is the cost scaling parameter, and $\mathcal{B} \subset (t_\beta, \infty)$ is the feasible region. This lower bound reflects the idea that a minimum level of interaction must be maintained to support essential activities and services. The cost term $c_\beta / (\beta - t_\beta)$ imposes increasing marginal costs as $\beta$ approaches $t_\beta$, capturing the intuition that further reductions in contact become progressively more difficult and disruptive as this lower bound is approached. 

We vary the cost parameters $c_v$ and $c_\beta$ and solve the resulting optimization problems via stochastic approximation (\cite{chau2014overview}). At each iteration, we average across $K=10$ replications of the gradient estimator using the sampling scheme described at the end of Section \ref{sec:parameter_uncertainty}. Fixing all other constants, we use recursions of the following form:
$$
v^{(n+1)} := \Pi_{\mathcal{V}}\left( v^{(n)} + \frac{A}{n+1} \left(\widehat J_{v, K}(v^{(n)}, \hat \beta, T) + c_v (1-v)^{-2}\right) \right),
$$
$$
\beta^{(n+1)} := \Pi_{\mathcal{B}}\left( \beta^{(n)} + \frac{A}{n+1} \left(\widehat{J}_{\beta, K}( \hat v, \beta^{(n)}, T) - c_\beta (\beta -t_{\beta})^{-2}\right) \right).
$$
where $\Pi$ represents projection onto the feasible region. 
For both problems, we perform 500 iterations of stochastic approximation for a given value of the cost parameter. We construct confidence intervals via 10 independent replications of the stochastic approximation procedure. 
As in previous simulations, we set $N=1000$, $i_0=1$, and $T=10$. We model parameter uncertainty in the same way as in Section \ref{sec:gradients_numerical}. We use an instance of the vaccination problem where we compare a fixed $\beta =5$ against uncertain $\hat \beta$ with shape parameter $a=10$ and mean $\beta_0=5$. For the contact problem, we consider a fixed $v=0.3$ versus uncertain $\hat v$ with a concentration parameter $C=10$ and $v_0=0.3$. The feasible regions are $\mathcal{V} = [0.0001, 0.9999]$ and $\mathcal{B} = [1.001, 1000]$. We use starting points of $v^{(0)} = 1 - 1/\beta = 0.8$ for the vaccination problem, and $\beta^{(0)} =1/(1-v) = 1.43$ for the contact problem. The lower bound in the contact problem is $t_\beta = 1$, and we use the same stepsize constant of $A = 1 \times 10^{-3}$ for both problems.

\begin{figure}
	\centering
	\includegraphics[width=0.46\linewidth]{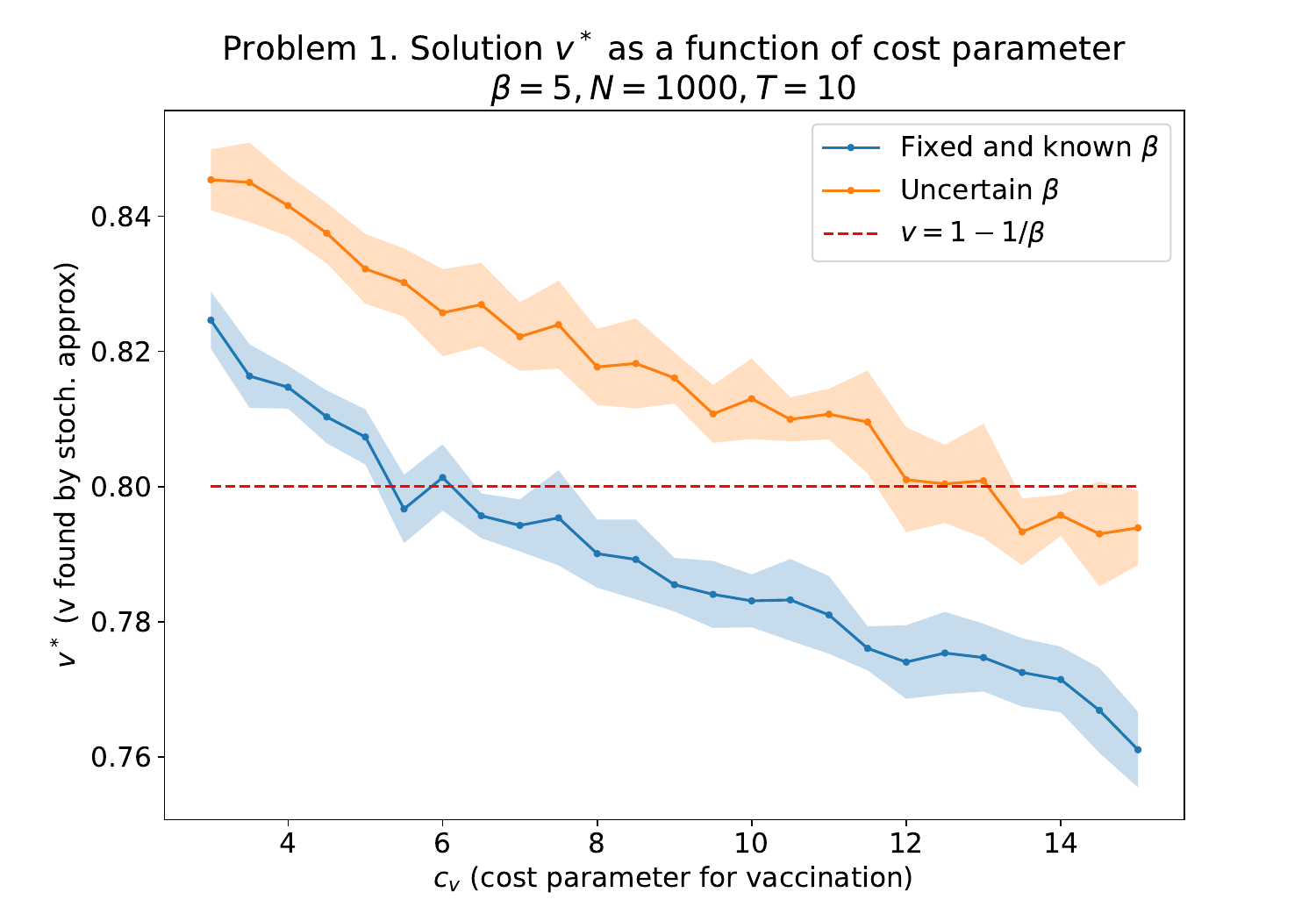}
	\includegraphics[width=0.46\linewidth]{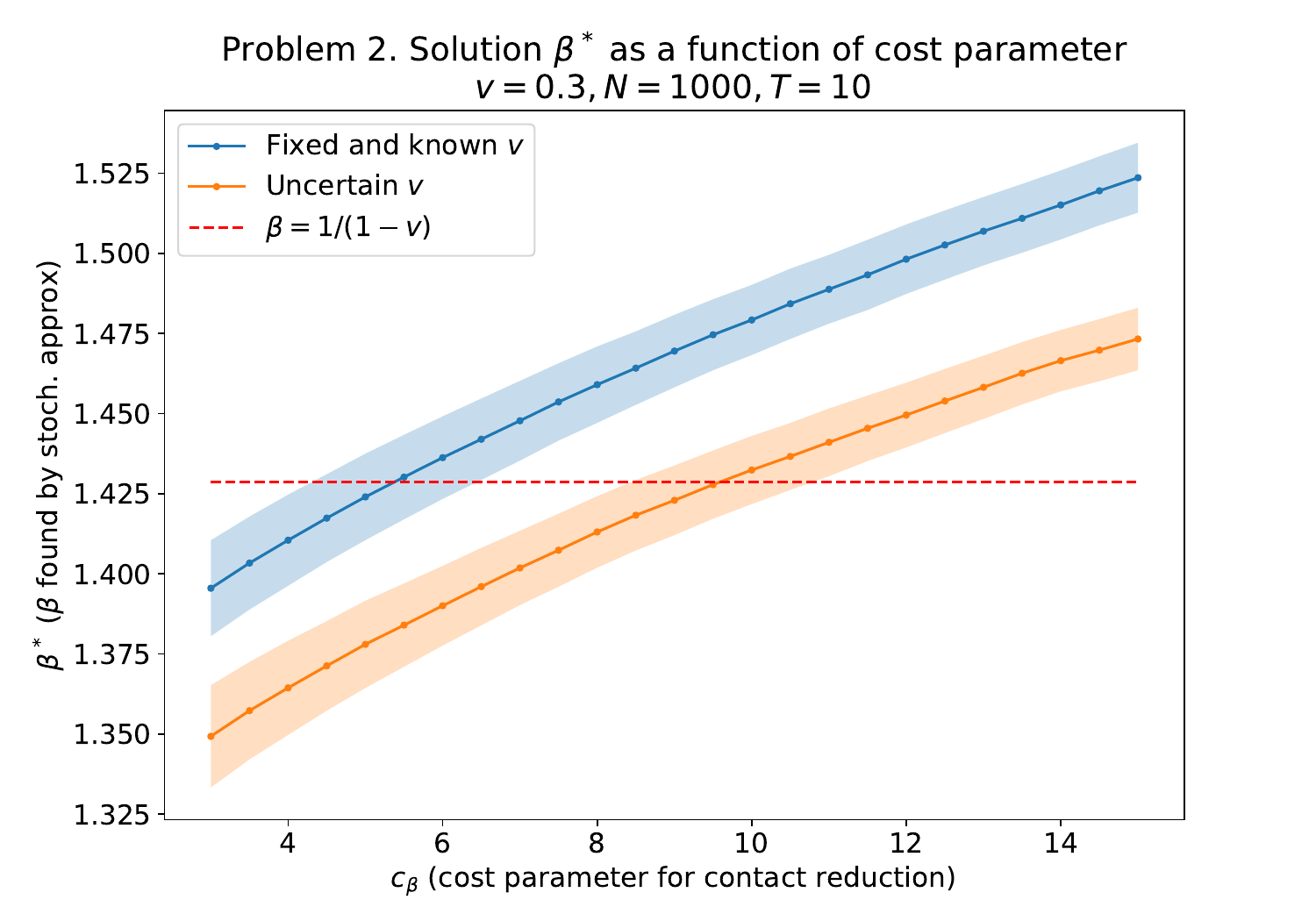}
	\caption{
	\textbf{Solutions to optimization problems as functions of intervention costs, with known vs. uncertain parameters, estimated via stochastic approximation.}
	The left panel shows the solution $v^*$ to the vaccination problem (\ref{eqn:opt_problem_1}) as a function of the cost parameter $c_v$. 
	The right panel shows the solution $\beta^*$ to the contact problem (\ref{eqn:opt_problem_2}) as a function of the cost parameter $c_\beta$. 
	In all plots, the shaded region is the 95\% CI from 10 independent replications of the stochastic approximation procedure, where each replication was run for 500 iterations.
}
	\label{fig:opt_soln}
\end{figure}

Figure \ref{fig:opt_soln} shows the optimization solutions as a function of the cost parameter.
We observe that the solution $v^*$ is a monotone decreasing function of the cost parameter $c_v$, and that a higher vaccination level is selected when the parameter $\beta$ is uncertain. Similarly, $\beta^*$ is a monotone increasing function of the cost parameter $c_\beta$, and a lower contact level is selected when $v$ is uncertain. The magnitude of the differences between the certain and uncertain solutions increase slightly as a function of the cost parameter, suggesting that nonlinearities in the problems increase the impact of parameter uncertainty.

These patterns are consistent with the intuition that more cautious decisions are required under uncertainty, leading to higher vaccination levels and lower contact rates. The vaccination level differs by as much as 2--4\%, which for $N=1000$ corresponds to vaccinating an additional 20--40 individuals. In the case of contact reduction, the optimal value of $\beta$ decreases by approximately 0.05 under uncertainty. Since the total number of contacts over the time horizon is on the order of $N \beta T$, this corresponds to roughly 500 fewer contacts, or a per capita reduction of approximately 0.5.

Each of the boundaries defined by the threshold conditions are shown as a dashed red line in the plots. We observe that the solutions without uncertainty cross the dashed red line at smaller values of the cost parameter than the solutions with uncertainty.
This crossing point can be interpreted as the value of the cost parameter at which it becomes optimal to accept $\mathcal{R}_{\text{eff}} > 1$ (i.e., nonzero probability of a large outbreak) rather than incur additional intervention costs. 
This observation indicates that the social planner should be less willing to make this tradeoff when parameters are uncertain.
Equivalently, this suggests that the willingness to pay for interventions is higher under uncertainty, as larger intervention costs are tolerated before allowing outcomes associated with $\mathcal{R}_{\text{eff}} > 1$.

\begin{figure}
	\centering
	\includegraphics[width=0.45\linewidth]{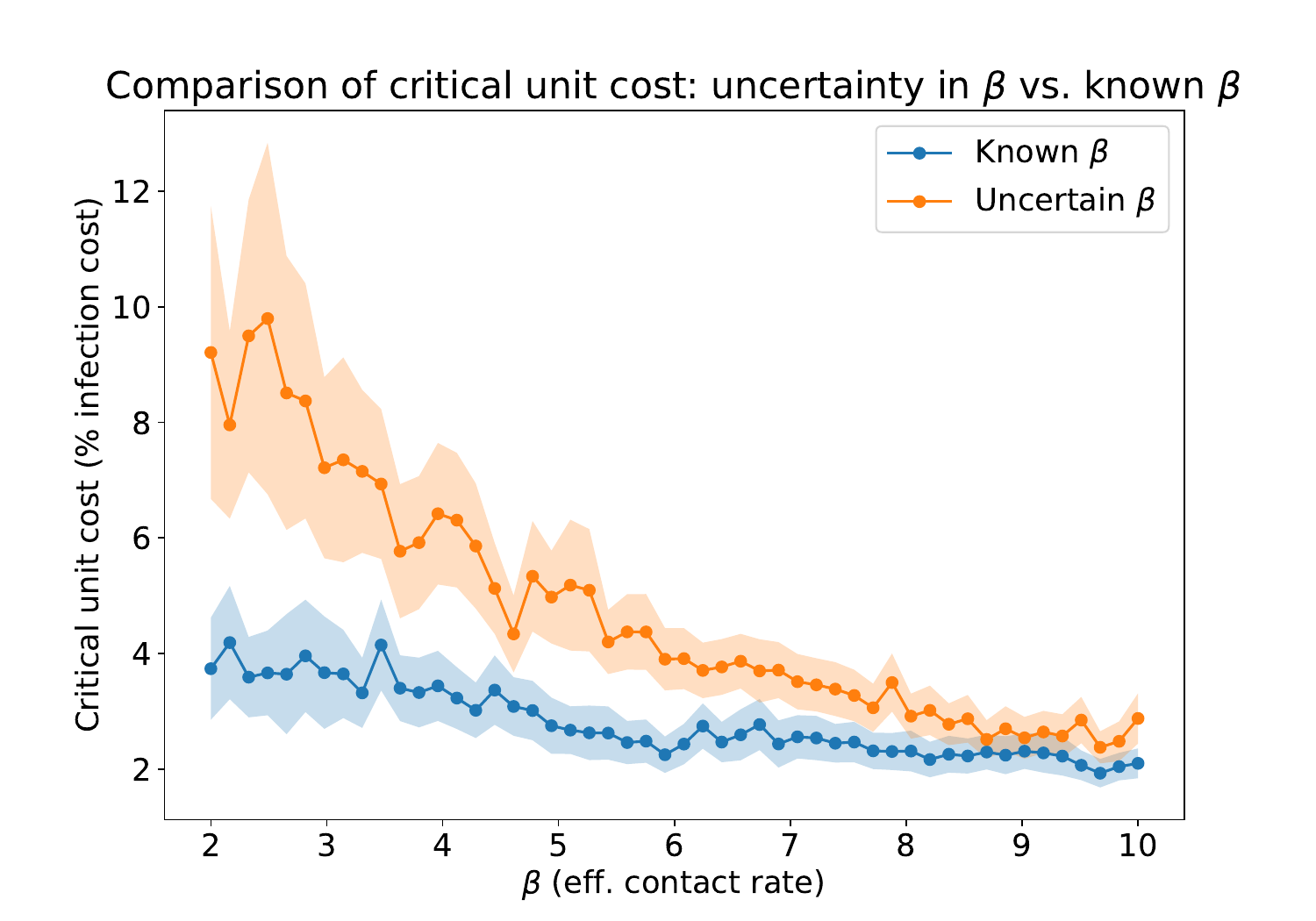}
	\includegraphics[width=0.45\linewidth]{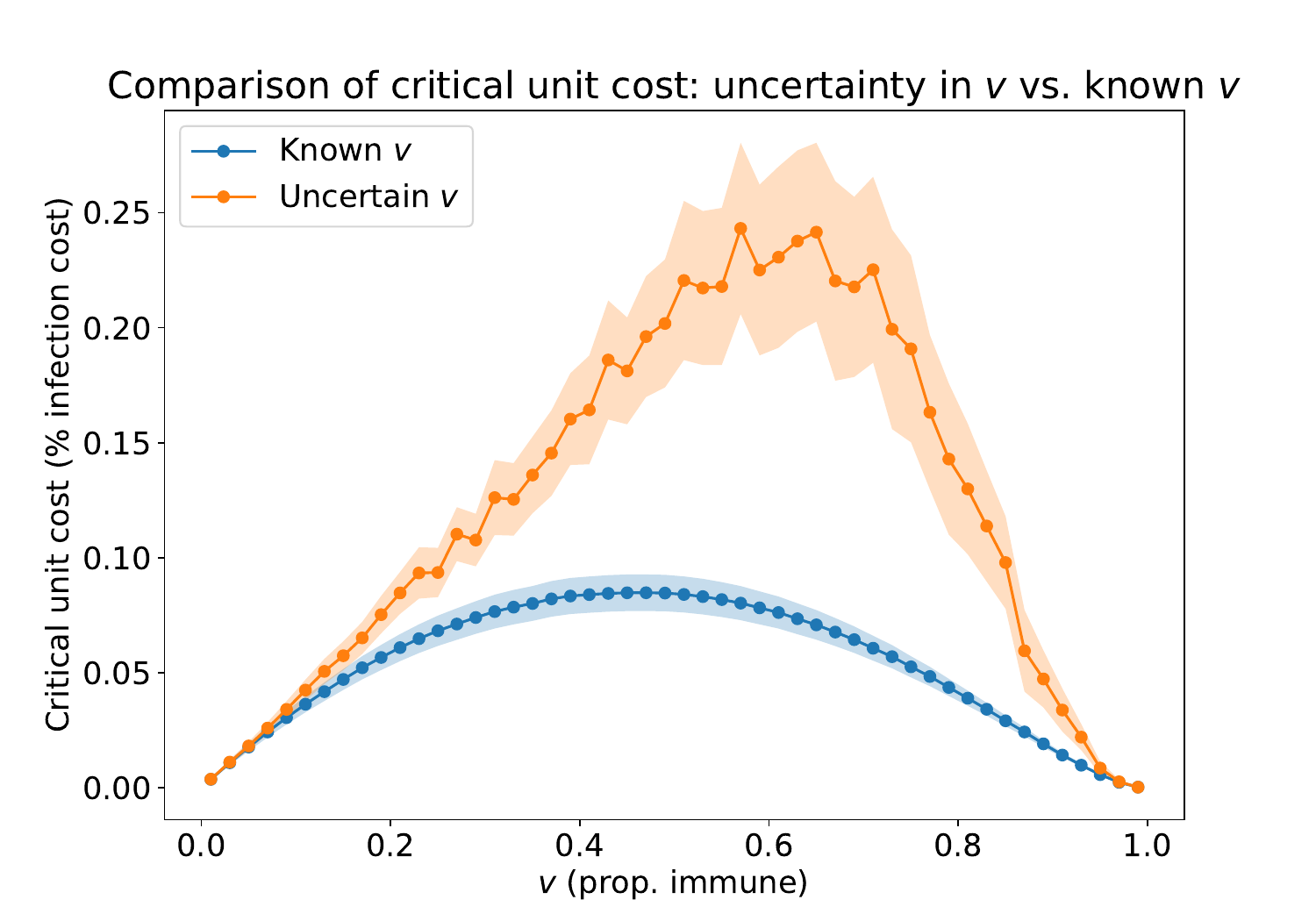}
	\caption{
		\textbf{Critical unit costs $C_v^*$ and $C_\beta^*$, as a function of $\beta$ and $v$, respectively. } The left plot shows estimates of $C_v^*$, corresponding to the vaccination problem. The right plot shows estimates of $C_\beta^*$, corresponding to the contact problem. The blue lines are the critical unit costs without parameter uncertainty, and the orange lines are the critical unit costs under parameter uncertainty, both estimated from 10,000 independent replications. Both plots have y-axes in units relative to the cost of one infection, expressed as a percentage. The shaded region is the 95\% CI, representing Monte-Carlo error.
	}
	\label{fig:critical_costs}
\end{figure}
	Motivated by this observation, we define the \textit{critical costs} $c_v^*$ and $c_\beta^*$ as the values of the cost parameters so that the equality $\mathcal{R}_{\text{eff}} = 1$ holds, given a fixed value of the other parameter. These are the largest values of the cost parameters that can be tolerated before allowing nonzero probability of a large outbreak. 
To find $c_v^*$ for a fixed $\beta$, we differentiate the objective in (\ref{eqn:opt_problem_1}) and set it equal to zero, set $v= 1-1/\beta$ (so that $\mathcal{R}_{\text{eff}}=1$), and solve for $c_v$. We can find $c_\beta^*$ in an identical manner: for a fixed $v$, we differentiate the objective in (\ref{eqn:opt_problem_2}), equate it to zero, set $\beta=(1-v)^{-1}$, and solve. 
In the general case where the parameters $v$ and $\beta$ are uncertain, and we have $\mathbb{E}[\hat \beta] = \beta_0$ and $\mathbb{E}[\hat v] = v_0$, we set 
$\beta = (1-v_0)^{-1}$ and $v = 1-1/\beta_0$, yielding the formulas:
$$
c_v^* = -\beta_0^{-2} \frac{d \mathbb{E}[Z(v,\hat \beta,T)]}{dv} \Big \lvert_{v = 1 - 1/\beta_0}, \quad
c_\beta^* = ((1-v_0)^{-1} - t_\beta)^2 \frac{d \mathbb{E}[Z(\hat v,\beta,T)]}{d\beta} \Big \lvert_{\beta = (1-v_0)^{-1}}.
$$
Thus, both $c_v^*$ and $c_\beta^*$ can be estimated by taking independent replications of the gradient estimators $\widehat J_v$ and $\widehat J_\beta$.
To give a cost-per-unit interpretation of the critical costs, we scale the total intervention cost by the magnitude of the intervention to obtain the \textit{critical unit cost}. For the vaccination problem, we divide the total cost of the intervention by the number of individuals vaccinated. If $v= 1-1/\beta$, then by plugging into the second term of (\ref{eqn:opt_problem_1}) we obtain a total cost of $c_v^*(\beta - 1)$, and a total of $Nv = N(1-1/\beta)$ individuals are immunized.  On the other hand, plugging $\beta = (1-v)^{-1}$ into the second term of (\ref{eqn:opt_problem_2}) yields a total cost of $c_\beta^*((1-v)^{-1} - t_\beta)^{-1}$, and the expected number of contacts over the time horizon is $N\beta T$. Thus, the critical unit costs are written as:
$$
C_v^* = -(\beta N)^{-1} \frac{d \mathbb{E}[Z(v,\beta,T)]}{dv} \Big \lvert_{v = 1 - 1/\beta},\quad
C_\beta^* = \frac{((1-v)^{-1} - t_\beta)}{NT(1-v)^{-1}} \frac{d \mathbb{E}[Z(v,\beta,T)]}{d\beta} \Big \lvert_{\beta = (1-v)^{-1}}.
$$
Since the cost terms are added directly to the expected number of infections, $C_v^*$ and $C_\beta^*$ are interpreted in units equivalent to the cost of a single infection.

	Figure \ref{fig:critical_costs} shows the critical unit costs estimated from 10,000 independent replications, as a function of $v$ and $\beta$, with $N=1000$, $i_0=1$, $T=10$. In both plots, the critical unit costs are larger under parametric uncertainty; thus, interventions become more valuable under uncertainty, as the social planner is willing to pay more per intervention. $C_v^*$ decreases as a function of $\beta$ and ranges from 4\% to 10\% of the cost of one infection, whereas $C_\beta^*$ is at most 0.25\% the cost of one infection: one vaccination is more valuable than the removal of one contact.
	For $C_v^*$, as $\beta$ increases, the vaccination threshold $v=1-1/\beta$ approaches 1, and the cost to vaccinate to this threshold increases without bound; the willingness of the social planner to pay for such an intervention goes to zero. 
	On the other hand, $C_\beta^*$ goes to zero as $v \rightarrow 0$ and $v \rightarrow 1$. 
	When $v \rightarrow 0$, $\beta = (1-v)^{-1} \rightarrow 1$, which corresponds to the lower bound $t_\beta$. In this case, the cost of reducing contacts to the threshold value explodes, and the willingness to pay for the intervention goes to zero.
	As $v \rightarrow 1$, $\beta = (1-v)^{-1} \rightarrow \infty$. In this regime, $\frac{d \mathbb{E}[Z(v,\beta,T)]}{d\beta}$ approaches 0, corresponding to the saturation effect discussed in Section \ref{sec:sensitivity_analysis}, and the social planner becomes unwilling to pay for an intervention that has no effect.

%% file: conclusion.tex
\vspace{-6pt}
\section{Conclusion}\label{s:conclusion}

In this work, we illustrated how gradient estimation methodologies can be applied to stochastic epidemic models to investigate how uncertainty impacts the effectiveness of interventions. We focused on incorporating the effects of randomness in transmission dynamics, as well as uncertainty in key epidemiological parameters. More broadly, this work provides a pathway for bridging predictive and prescriptive techniques in stochastic epidemic modeling, by showing how to incorporate uncertainties from model calibration into optimization and sensitivity analyses via gradient estimation.

We derived unbiased gradient estimators for the total number of infections with respect to the proportion immunized and the effective contact rate for an approximate discrete-time epidemic model. We used the weak-derivative approach to derive the gradient estimator with respect to the proportion immunized, and we used the likelihood-ratio method to derive the gradient estimator with respect to the effective contact rate. 
After applying variance reduction techniques such as CRN, the estimators achieve variance comparable to biased one-sided finite-difference estimators with CRN, while remaining unbiased. The likelihood-ratio estimator offered computational advantages, and the weak-derivative estimator had an attractive theoretical interpretation.
However, the gradient estimators differed from results using deterministic limit formulas, as the estimators consider the case of finite populations and time horizons. 

After extending the gradient estimators to accommodate uncertainty in the parameters, we conducted two illustrative simulation studies. 
First, we studied how parameter uncertainty impacted the marginal effect of vaccinations or the contact rate. The numerical results indicate that additional uncertainty decreases the marginal effects of vaccination or contact reduction. Second, we applied the gradient estimators to two simple optimization problems, involving tradeoffs between the costs of intervention and the costs of infections. 
Under parameter uncertainty, the optimal solutions shifted toward stronger interventions, reflecting a higher willingness to pay for outbreak prevention.

There are several directions of future work to be considered. The gradient estimation approach could be extended to the continuous-time case. The epidemic model could incorporate additional heterogeneities like variation in contact rates (e.g., contacts between cities versus contacts within cities). Another line of work concerns risk-sensitive optimization. Consequential outcomes for an epidemic lie in the tail of the outbreak size distribution, motivating the use of risk measures such as Conditional-Value-at-Risk (CVaR).